\documentclass[10pt,conference]{IEEEtran}
\usepackage{amsmath,amssymb}
\usepackage{graphicx}
\usepackage{mathrsfs}
\usepackage{tikz}
\usepackage{booktabs}
\usepackage{floatflt}
\usepackage{enumerate}
\usepackage{psfrag}	
\usepackage{array}
\usepackage{multirow,hhline}
\usepackage{exscale}
\usepackage{color}
\usepackage{colortbl}
\usepackage{epsfig}
\usepackage{cite}
\usepackage{amsthm}
\newcommand{\mat}[1]{\ensuremath{\mathbf{#1}}}
\renewcommand{\vec}[1]{\ensuremath{\mathbf{#1}}}
\renewcommand{\r}[1]{\ensuremath{\text{rank}({#1})}}
\newcommand{\s}[1]{\ensuremath{\text{span}({#1})}}
\newtheorem{theorem}{Theorem}
\newtheorem{lemma}{Lemma}
\newtheorem{remark}{Remark}
\newtheorem{corollary}{Corollary}

\begin{document}

\IEEEoverridecommandlockouts
\title{The DoF of the K-user Interference Channel with a Cognitive Relay}
\author{
\IEEEauthorblockN{Anas Chaaban and Aydin Sezgin}
\IEEEauthorblockA{Chair of Communication Systems\\
RUB, 44780 Bochum, Germany\\
Email: {anas.chaaban@rub.de, aydin.sezgin@rub.de}}
}

\maketitle

\begin{abstract}
It was shown recently that the $2$-user interference channel with a cognitive relay (IC-CR) has full degrees of freedom (DoF) almost surely, that is, $2$ DoF. The purpose of this work is to check whether the DoF of the $K$-user IC-CR, consisting of $K$ user pairs and a cognitive relay, follow as a straight forward extension of the $2$-user case. As it turns out, this is not the case. The $K$-user IC-CR is shown to have $2K/3$ DoF if $K>2$ for the when the channel is time varying, achievable using interference alignment. Thus, while the basic $K$-user IC with time varying channel coefficients has $1/2$ DoF per user for all $K$, the $K$-user IC-CR with varying channels has 1 DoF per user if $K=2$ and $2/3$ DoF per user if $K>2$. 
Furthermore, the DoF region of the $3$-user IC-CR with constant channels is characterized using interference neutralization, and a new upper bound on the sum-capacity of the $2$-user IC-CR is given. 
\end{abstract}

\section{Introduction}
One of the approaches for approximating the capacity of interference networks is finding the multiplexing gain. The multiplexing gain, also known as the capacity pre-log or degrees of freedom (DoF), characterizes the capacity of the network at an asymptotically high signal-to-noise ratio. Recently, there has been an increasing interest in characterizing the DoF of interference networks, e.g, the $K$-user IC \cite{CadambeJafar_KUserIC} and the X Channel \cite{JafarShamai_XChannel,MaddahAliMotahariKhandani_XChannel}.

Besides the IC and the X channel, relaying setups have also been studied from DoF point of view. For instance, \cite{CadambeJafar_ImpactOfRelays} studies the impact of relays on wireless networks and shows that causal relays do not increase the DoF of the network. Non-causal relays, on the other hand, can increase the DoF. In \cite{SridharanVishwanathJafarShamai}, achievable rate regions and upper bounds for the 2-user IC with a cognitive relay (IC-CR) were given, and 
it was shown the interference channel with a cognitive relay has full DoF, i.e., 2 DoF. The cognitive IC has also been studied in \cite{RiniTuninettiDevroyeGoldsmith,RiniTuninettiDevroye_ISIT, RiniTuninettiDevroye_ITW} where capacity results for some cases were given, in addition to new upper and lower bounds.

The question we try to answer in this paper is: How is the behavior of the DoF of the IC-CR with $K$ users? A straight forward extension of the results of \cite{SridharanVishwanathJafarShamai} suggest that the $K$-user IC-CR has $K$ DoF. The goal of this paper is the characterization of the DoF for general $K$. Namely, we consider the effect of a cognitive relay on the DoF of the $K$-user IC. We study the $K$-user Gaussian IC-CR, and obtain the DoF of this channel under time varying channel coefficients assumption.

It turns out that the case with $K>2$ users does not follow as a straight forward extension of the 2-user case. We show that while the sum-rate of the 2-user Gaussian IC-CR scales as $2\log(P)$ as the transmit power $P\to\infty$, the $K>2$ user case scales as $\frac{2K}{3}\log(P)$. In other words, the 2-user case does not follow the same law as the $K>2$ user case. This DoF is shown to be achievable using interference alignment as in a $K$-user $2\times1$ MISO IC \cite{GouJafar}. Thus we give a characterization of the DoF of the $K$-user Gaussian IC-CR with time varying channel coefficients. It turns out that the per user DoF of the $K$-user Gaussian IC-CR drop from 1 to $2/3$ as we go from the $K=2$ to $K>2$. This is in contrast to the $K$-user IC, where the per-user DoF is $1/2$ for all $K\geq2$. We also consider the constant channel case, for which we obtain the DoF region of the 3-user Gaussian IC-CR.

As a result, in contrast to \cite{CadambeJafar_ImpactOfRelays}, where it was shown that causal relays can not increase the DoF of the wireless network, a cognitive relay can increase the DoF of the $K$-user IC from $K/2$ to $2K/3$ with $K>2$. Moreover, the results of this paper give an example where cognition/relaying can help in increasing the DoF of a wireless network.

The rest of the paper is organized as follows. In section \ref{Model}, we give the general model of the $K$-user IC-CR. In section \ref{TimeVarying}, we consider the time varying IC-CR and characterize its DoF, and in section \ref{Constant} we consider the IC-CR with constant channel coefficients, where we give a new sum-rate upper bound for the 2-user case and characterize the DoF region of the 3-user case. Finally, we conclude in section \ref{Conc}.

\section{System Model}
\label{Model}
The K-user Gaussian interference channel with a cognitive relay (IC-CR) is shown in Figure \ref{KuserICCR}. It consists of $K$ transmit-receive pairs and a cognitive relay, each with one antenna. For $k\in\{1,\dots,K\}$, source $k$ has a message $m_k\in\mathcal{M}_k\triangleq\{1,\dots,2^{nR_k}\}$ to be sent to destination $k$ over $n$ channel uses. The messages $m_k$ are independent, uniformely distributed over the messages sets, and are made available non-causally at the relay. At each time instant $(i)$, the output of the channel can be represented as follows
\begin{eqnarray*}
Y_{k}(i)=\sum_{j=1}^Kh_{jk}(i)X_{j}(i)+h_{rk}(i)X_{r}(i)+Z_{k}(i),
\end{eqnarray*}
where $X_r,X_k\in\mathbb{R}$, $k=1,\dots,K$, are the channel inputs and $Y_k\in\mathbb{R}$ is the channel output, $Z_k$ is an independent and identically distributed (i.i.d.) noise with zero mean and unit variance $Z_k\sim\mathcal{N}(0,1)$, and $h_{jk}(i)$ and $h_{rk}(i)$ represent time varying channel gains from source $j$ and the relay to destination $k$, respectively. The channels are assumed to be known apriori at all nodes, and are i.i.d. and drawn from a continuous distribution. The IC-CR with constant channels is defined in the same way as above, with the exception that $h_{jk}(i_1)=h_{jk}(i_2)$ and $h_{rk}(i_1)=h_{rk}(i_2)$ for all $i_1,i_2\in\mathbb{N}$.

The inputs satisfy the following power constraint
\begin{equation}
\label{PowerConstraint}
% \frac{1}{n}\sum_{i=1}^n\mathbb{E}[X_{j,i}^2]\leq P,\quad \forall j\in\{1,\dots,K,r\}.
\mathbb{E}[X_{j}^2]\leq P,\quad \forall j\in\{1,\dots,K,r\}.
\end{equation}
The transmitters and the relay use encoding functions 
% \begin{align}
% f_k&:\mathcal{M}_k\to\mathbb{R}^n,\ k\in\{1,\dots,K\}\\
% f_r&:\mathcal{M}_1\times\dots\times\mathcal{M}_K\to\mathbb{R}^n,
% \end{align}
to map the messages to codewords $X_k^n=(X_k(1),\dots,X_k(n))$ and $X_r^n=(X_r(1)\dots,X_r(n))$, respectively. The receivers want to decode their desired messages from their received signals $Y_k^n$
% transmitted use and decoding functions 
% \begin{align}
% g_k&:\mathbb{R}^n\to\mathcal{M}_k,\ k\in\{1,\dots,K\}
% \end{align}
which induces an error probability. 
%such that the transmit codewords $$X_j^n\triangleq\left\{X_{j,1},\dots,X_{j,n}\right\},\quad j\in\{1,\dots,K,r\}$$ satisfy the power constraint (\ref{PowerConstraint}) and the decoding error probability $P_e\to0$ as $n\to\infty$. 
A rate tuple $(R_1,\dots,R_K)$ is said to be achievable if the error probability can be made arbitrarily small by increasing the code length $n$. The closure of the set of all achievable rate tuples defines the capacity region $\mathcal{C}$.

% The closure of all achievable rates defines the rate region $\mathcal{R}$. The DoF region is defined as the set $$\left\{(d_1,\dots,d_K): d_i=\lim_{P\to\infty}\frac{R_i}{0.5\log(P)}, R_i\in\mathcal{R} \forall i\in\{1,\dots,K\}\right\}.$$

An achievable sum-rate is defined as $R_\Sigma=\sum_{k=1}^KR_k$ with $(R_1,\dots,R_k)\in\mathcal{C}$ and the sum-capacity $C_\Sigma$ is the maximum sum-rate. The sum DoF is defined as $$d_\Sigma=\sum_{k=1}^Kd_i=\lim_{P\to\infty}\frac{C_\Sigma(P)}{\frac{1}{2}\log(P)}.$$
The DoF region $\mathcal{D}$ is defined as in \cite{CadambeJafar_KUserIC}.
% $${\mathcal{D}=\left\{\begin{array}{l}
% (d_1,\dots,d_K)\in\mathbb{R}_+^K:\forall w_1,\dots,w_K\in\mathbb{R}_+\\
% \sum_{k=1}^K w_kd_k\leq\lim\sup_{P\to\infty}\\
% \hspace{2.8cm}\sup_{(R_1,\dots,R_K)\in\mathcal{C}}\frac{\sum_{k=1}^Kw_kR_k}{\log(P)}\end{array}\right\}.}
% $$

\begin{figure}
\centering
\includegraphics[width=0.9\columnwidth]{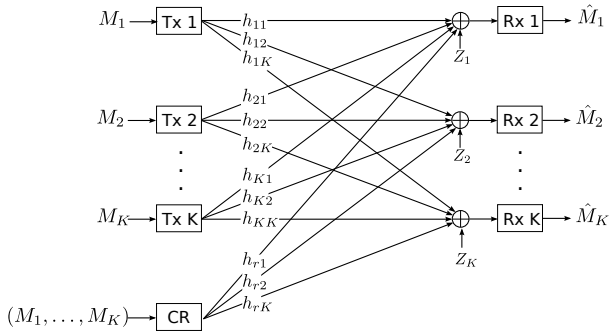}
\caption{The K-user Gaussian interference channel with a cognitive relay (CR) system model.}
\label{KuserICCR}
\end{figure}

\section{The IC-CR with Time Varying Channel Coefficients}
\label{TimeVarying}
In this section, we study the DoF of the $K$-user IC-CR. We state the main result in the following theorem, and describe it in more details afterwards.
\begin{theorem}
The DoF of the $K$-user IC-CR with time varying channel coefficients is given by
\begin{align*}
d_\Sigma=\left\{\begin{array}{cc} 2 & \text{if $K=2$}\\
\frac{2K}{3} & \text{if $K>2$}\end{array}\right.
\end{align*}
\end{theorem}
The proof of this theorem is given in the following subsections. We consider the 2-user case first, and then the $K$-user case, and derive upper bounds on the DoF. Then we provide the achievability of these upper bounds.

\subsection{A Sum-capacity Upper Bound for the 2-User Gaussian IC-CR}
\label{2UserSumRateBound}
The 2-User Gaussian IC-CR with constant channel coefficients was considered in \cite{SridharanVishwanathJafarShamai}, where achievable rate regions, upper bounds, and the DoF region were given. The same DoF upper bound as in \cite{SridharanVishwanathJafarShamai} holds for the time varying case. That is
% \begin{theorem}
% \label{2UserICCRTheorem}
% The sum DoF of the $2$-user Gaussian IC-CR with time varying channel coefficients is upper bounded by
\begin{align}
\label{2UserDoFUpperBoundTimeVarying}
d_1+d_2\leq2.
\end{align}

\subsection{DoF Upper Bound for the $K$-User Gaussian IC-CR with $K\geq3$:}
\label{3User}
We first consider the case $K=3$. The DoF upper bound (\ref{2UserDoFUpperBoundTimeVarying}) yields
\begin{equation}
\label{DoFExtension}
d_1+d_2+d_3\leq3,
\end{equation}
when extended to the $3$-user case. However, as we show next, this straight forward extension is not tight since the DoF of the 3-user IC-CR is upper bounded by 2. In the following lemma, we give a DoF upper bound for the 3-user IC-CR.
\begin{lemma}
\label{3UserICCRTheorem}
The DoF of the 3-user IC-CR is upper bounded by
\begin{equation*}
d_1+d_2+d_3\leq2.
\end{equation*}
\end{lemma}
\begin{proof}
See Appendix \ref{3UserICCRProof}.
\end{proof}

The $3$-user Gaussian IC-CR DoF upper bound can be used to obtain the DoF upper bound for the $K$-user Gaussian IC-CR with $K\geq3$ stated in the following theorem.
\begin{theorem}
The DoF of the $K$-user IC-CR, $K\geq3$, is upper bounded as follows 
\begin{equation}
\label{KUserDoFUpperBoundTimeVarying}
d_\Sigma\leq\frac{2K}{3}.
\end{equation}
\end{theorem}
\begin{proof}
Using Lemma \ref{3UserICCRTheorem}, we have: $d_j+d_k+d_l\leq 2$, for all distinct $j,k,l\in\{1,\dots,K\}$. Adding all such inequalities, we obtain ${{K-1}\choose{2}}d_\Sigma\leq2{{K}\choose{3}}$, and the result follows.
\end{proof}

\subsection{Achievability of the $K$-User IC-CR DoF}
Consider the following achievable scheme in a $K$-user Gaussian IC-CR. At time instant $i$, the message $m_k$, $k\in\{1,\dots,K\}$, is mapped to a vector $\vec{x}_k(i)=[x_{k}^{[1]}(i),\ x_{k}^{[2]}(i)]^T$, the first component of which is sent from Tx 1 and the second component is sent from the relay. The overall relay signal is $x_r(i)=\sum_{k=1}^Kx_k^{[2]}(i)$ and the received signals at receiver $j$ can be written as:
\begin{align}
\label{MISOIO}
y_{j}(i)&=\sum_{k=1}^K\vec{h}_{kj}^T(i)\vec{x}_{k}(i)+z_{j}(i),\\
\label{HMISO}
\vec{h}_{kj}(i)&=[h_{kj}(i),\ h_{rj}(i)]^T. 
\end{align}
Therefore, we can model the IC-CR with this scheme as a $K$-user $2\times1$ MISO IC with time varying channel coefficients. Since the relay sends the sum of $K$ signals, we guarantee that the power constraint at each node of the IC-CR is satisfied by defining the power constraint of the resulting MISO channel to be $P/K$ at each node. Notice that this power scaling does not reduce the achievable DoF.

It was shown in \cite{GouJafar} that using interference alignment in a $K$-user $2\times1$ MISO interference channel with time varying channel coefficients, 2 DoF are achievable if $K=2$, and $\frac{2K}{3}$ DoF are achievable if $K>2$. It is important to note that by the reciprocity of interference alignment \cite{GomadamCadambeJafar_GC2008}, the same DoF is achievable in the SIMO IC (with the same physical channels). 

Here, we use the same scheme as in \cite{GouJafar} for our setup, i.e., we make use of reciprocity. We consider the reciprocal $1\times2$ $K$-user SIMO IC with the physical channels given by the $2\times1$ MISO IC interpretation of the IC-CR given in \eqref{MISOIO}. In this SIMO IC, the channel from transmitter $j$ to receiver $k$ is $\vec{h}_{kj}(i)$. Notice here the special structure of the SIMO channel vectors: the second component of $\vec{h}_{jj}(i)$ is the same as $\vec{h}_{kj}(i)$ (see \eqref{HMISO})

Now, as in \cite{GouJafar}, we consider $\mu_n$ symbol extensions of the channel. This makes the $1\times2$ SIMO IC and extended $\mu_n\times2\mu_n$ SIMO IC, where the channel matrix from Tx $j$ to Rx $k$ is $2\mu_n\times\mu_n$ and has a block diagonal structure
\begin{align}
H_{kj}=\left[\begin{array}{ccc} \vec{h}_{kj}(1)&\vec{0}_{2\times1}&\dots\\
\vec{0}_{2\times1}&\vec{h}_{kj}(2)&\dots\\
\vdots&\vdots&\ddots\end{array}\right],
\end{align}
where $\vec{0}_{2\times1}$ is the all-zero vector of length 2. Notice that $\vec{h}_{kj}(i_2)$ and $\vec{h}_{kj}(i_2)$ are independent. User $j\in\mathcal{T}_1=\{1,2,3\}$ sends a data vector $\vec{w}_j=[\vec{x}_j^T(1),\ \vec{x}_j^T(2),\dots]^T$ of length\footnote{$\mu_n$ is chosen so that all the relevant quantities are integer.} $\frac{2}{3}\mu_n$ using a pre-coding matrix $V_j$ with dimension $\mu_n\times\frac{2}{3}\mu_n$. User $j\in\mathcal{T}_2=\{4,\dots,K\}$ sends a data vector $\vec{w}_j$ of length $\left(\frac{2}{3}-\epsilon_n\right)\mu_n$ using a pre-coding matrix $V_j$ with dimension $\mu_n\times(\frac{2}{3}-\epsilon_n)\mu_n$, where $\epsilon_n\to0$ as $n\to\infty$. Thus, Tx $j$ sends 
\begin{align}
X_j=V_j\vec{w}_j.
\end{align}
As in \cite{GouJafar}, we choose $V_1=V_2=V_3$ and $V_4=V_5=\dots=V_K$. The main idea is of alignment is to find pre-coding matrices $V_k$ and post-coding matrices $U_k$ such that 
\begin{align}
\label{IALI}
\r{U_kH_{kk}V_k}&=d_k\\
\label{IA}
U_kH_{kj}V_j&=0\quad\forall k\neq j,
\end{align}
where $d_k=\frac{2}{3}\mu_n$ for $k\in\mathcal{T}_1$ and $d_k=(\frac{2}{3}-\epsilon_n)\mu_n$ for $k\in\mathcal{T}_2$. Here, $d_k$ denotes the dimension of the subspace spanned by the desired signal at Rx $k$. Denote by $\bar{d}_k$ the dimension of the subspace spanned by all the interfering signals arriving at Rx $k$. Since user $k$ needs to achieve $d_k$ DoF, then the remaining dimensions of the overall $2\mu_n$-dimenstional receive space to be occupied by interference should have $\bar{d}_k=\frac{4}{3}\mu_n$ for $k\in\mathcal{T}_1$ and $\bar{d}_k=(\frac{4}{3}+\epsilon_n)\mu_n$ for $k\in\mathcal{T}_2$. For example, at Rx 1 and $K$, we need to make sure that the following holds, respectively, 
\begin{align*}
\s{[H_{12}V_2,\ H_{13}V_3,\dots,H_{1K}V_K]}&=\frac{4}{3}\mu_n\\
\s{[H_{K1}V_1,\ H_{K2}V_2,\dots,H_{K(K-1)}V_{K-1}]}&=(\frac{4}{3}+\epsilon_n)\mu_n.\end{align*}
This is guaranteed by using the same construction of $V_k$ as in \cite{GouJafar}, where $V_k$ is given as a function of all $H_{kj}$, $j\neq k$. By choosing $U_k$ to be the null space of the subspace spanned by the interference, we satisfy \eqref{IA}.

Now for the general SIMO IC, the construction of $V_k$ given in \cite{GouJafar} also satisfies \eqref{IALI} since their design of $V_k$ is independent of the direct channels $H_{kk}$ which are generated randomly and independently of all other channels. In our case, we should examine this more carefully, since we have some dependency in the channels given by $$H_{kk}(2m,m)=H_{kj}(2m,m),\ m=1,2,\dots$$ where $H_{kj}(a,b)$ is the component in the a-th row and b-th column of $H_{kj}$. The design of $V_k$ is not completely independent of $H_{kk}$ in our case. However, let us write $H_{kk}$ as follows
\begin{align}
H_{kk}=\widehat{H}_{kk}+\widetilde{H}_{kk}
\end{align}
where
\begin{align}
\widehat{H}_{kk}&=\left[\begin{array}{ccc} h_{kk}(1)&0&\dots\\
0&0&\dots\\
0&h_{kk}(2)&\dots\\
0&0&\dots\\
\vdots&\vdots&\ddots\end{array}\right].%\\
% \widetilde{H}_{kk}&=\left[\begin{array}{ccc} 0&0&\dots\\
% h_{rk}(1)&0&\dots\\
% 0&0&\dots\\
% 0&h_{rk}(2)&\dots\\
% \vdots&\vdots&\ddots\end{array}\right].
\end{align}
Then, the construction of $V_k$ is clearly independent of $\widehat{H}_{kk}$ whose components are independent of all other channel matrices. Moreover, $\widehat{H}_{kk}$ has full rank. Therefore, $\r{U_k\widehat{H}_{kk}V_k}=d_k$ almost surely and hence condition \eqref{IALI} is satisfied. This achieves $3\left(\frac{2}{3}\right)\mu_n+(K-3)\left(\frac{2}{3}-\epsilon_n\right)\mu_n$ DoF almost surely which approaches $\frac{2K}{3}$ as $n\to\infty$. As a consequence (due to reciprocity), by using $V_j$ and $U_j$ as post-coding and pre-coding matrices at Rx $j$ and Tx $j$ in the original MISO IC, respectively, we achieve $2K/3$ DoF. Thus the DoF upper bounds (\ref{2UserDoFUpperBoundTimeVarying}) and (\ref{KUserDoFUpperBoundTimeVarying}) are achievable using interference alignment.

% As a result, we can state the following theorem.
% \begin{theorem}
% The $K$-user Gaussian IC-CR with time varying channel coefficients has $\frac{2K}{3}$ DoF almost surely$$\sum_{k=1}^Kd_k=\frac{2K}{3},\quad K\geq3,$$ achievable via interference alignment.
% \end{theorem}

\section{The IC-CR with constant channel coefficients}
\label{Constant}
In this section, we focus on the IC-CR with constant channel coefficients. We give a new sum-rate upper bound for the 2-user IC-CR. The DoF upper bounds in section \ref{TimeVarying} are general and still hold in this case. However, what differs is that achievable scheme. In what follows, we give an upper bound on the sum-rate of the 2-user case, and we characterize the DoF region of the 3-user case.

\subsection{The 2-User Gaussian IC-CR with constant channel coefficients}
\label{2User}

\begin{theorem}
\label{2UserICCRTheoremConstant}
The sum-rate of the 2-user Gaussian IC-CR with constant channel coefficients is upper bounded by
\begin{align*}
&R_1+R_2\leq\max_{\mat{A}\succeq0}\left\{I(X_{1},X_{2},X_{r};Y_{1})+I(X_{2},X_{r};Y_{2}|Y_{1},X_{1})\right\}\nonumber
\end{align*}
where $(X_1,X_2,X_r)$ are jointly Gaussian with covariance matrix
\begin{equation*}
% \label{CovMat}
\mat{A}=\left(\begin{array}{ccc}
P_1 	& 0 	& \rho_1\sqrt{P_1P_r}\\
0   	& P_2	& \rho_2\sqrt{P_2P_r}\\
\rho_1\sqrt{P_1P_r} & \rho_2\sqrt{P_2P_r} & P_r\end{array}\right),
\end{equation*}
and $P_j\leq P\ \forall j\in\{1,2,r\}$. This bounds gives the following DoF upper bound
\begin{eqnarray}
\label{2UserDoF}
d_1+d_2=\left\{\begin{array}{cll} 1 & \text{if } &\hspace{-0.5cm}h_{11}h_{r2}-h_{12}h_{r1}=0\\
& \text{or } &\hspace{-0.5cm}h_{22}h_{r1}-h_{21}h_{r2}=0\\
2 & \text{otherwise}
\end{array}
\right.
\end{eqnarray}
\end{theorem}
The statement of this theorem is obtained by giving $(Y_1^n,m_1)$  as side information to receiver 2 and using classical information theoretic approaches. % We do not include the detailed proof here due to space constraints. This upper bound can be expressed as follows
% \begin{eqnarray}
% \label{2UserRsum}
% R_1+R_2\leq\max_{\mat{A}\succeq0}\left\{C(Q_1)+C(Q_2)\right\}.
% \end{eqnarray}
% where $C(x)=\frac{1}{2}\log(1+x)$, and
% \begin{eqnarray*}
% Q_1&=&h_{11}^2P_1+h_{21}^2P_2+h_{r1}^2P_r\\
% &&+2h_{11}h_{r1}\rho_1\sqrt{P_1P_r}+2h_{21}h_{r1}\rho_2\sqrt{P_2P_r}\\
% Q_2&=&\frac{(h_{22}h_{r1}-h_{21}h_{r2})^2P_2P_r}{1+h_{21}^2P_2+h_{r1}^2P_r+2h_{21}h_{r1}\rho_2\sqrt{P_2P_r}}\\
% &&+\frac{h_{22}^2P_2+h_{r2}^2P_r+2h_{22}h_{r2}\rho_2\sqrt{P_2P_r}}{1+h_{21}^2P_2+h_{r1}^2P_r+2h_{21}h_{r1}\rho_2\sqrt{P_2P_r}},
% \end{eqnarray*}
% and thus, leads to the following DoF upper bound
% \begin{align}
% \label{2UserDoFUpperBound}
% d_1+d_2\leq\left\{\begin{array}{cl}1 & \text{if } h_{22}h_{r1}-h_{21}h_{r2}=0\\2 & \text{otherwise}\end{array}\right.
% \end{align}
% Similarly we can obtain a similar bound by replacing 1(2) in the expressions of $Q_1$ and $Q_2$ with 2(1). Thus if $h_{11}h_{r2}-h_{12}h_{r1}\neq0$, then $d_1+d_2\leq2,$ otherwise $d_1+d_2\leq1.$
% \begin{corollary}
% The sum DoF of the 2-user Gaussian IC-CR with constant channel coefficients are given by
% \begin{eqnarray}
% \label{2UserDoF}
% d_1+d_2=\left\{\begin{array}{cll} 1 & \text{if } &\hspace{-0.5cm}h_{11}h_{r2}-h_{12}h_{r1}=0\\
% & \text{or } &\hspace{-0.5cm}h_{22}h_{r1}-h_{21}h_{r2}=0\\
% 2 & \text{otherwise}
% \end{array}
% \right.
% \end{eqnarray}
% \end{corollary}
In \cite[Theorem 4]{SridharanVishwanathJafarShamai}, it was shown that $d_1+d_2$ satisfies \eqref{2UserDoF}, and that this upper bound is indeed achievable using interference neutralization \cite{MohajerDiggaviFragouliTse}. The sum-rate upper bound in Theorem \ref{2UserICCRTheoremConstant} combines the two DoF cases in one expression. We notice a collapse of the DoF to 1 under the special conditions in (\ref{2UserDoF}). With random channel realizations, the condition under which $d_1+d_2=1$ constitutes a set of measure zero. Thus the 2-user IC-CR with constant channel coefficients has 2 DoF almost surely achievable using interference neutralization.

\subsection{The 3-User Gaussian IC-CR with constant channel coefficients}
The sum DoF upper bound in Lemma \ref{3UserICCRTheorem} still holds in this case. Thus $$d_1+d_2+d_3\leq2.$$
In the following theorem, we give the DoF region of the 3-user IC-CR with constant channel coefficients.
\begin{theorem}
The DoF region $\mathcal{D}$ of the 3-user Gaussian IC-CR is given by
\begin{align}
\label{3UserICCRRegion}
\mathcal{D}=\left\{(d_1,d_2,d_3):\begin{array}{l}d_k\leq1, \forall k\in\{1,2,3\}\\
d_1+d_2+d_3\leq2\end{array}
\right\}.
\end{align}
\end{theorem}
\begin{proof}
We know that $d_1+d_2+d_3\leq2$. Together with the following trivial bounds $$d_k\leq1, \forall k\in\{1,2,3\},$$ it follows that the DoF region is outer bounded by $\mathcal{D}$. Since the corner points of this region, i.e. the points $(1,0,0)$,$(0,1,0)$, and $(0,0,1)$, $(1,1,0)$, $(1,0,1)$, and $(0,1,1)$ are all achievable, the former three corners by keeping two users silent, and the latter three corners by keeping one user silent and using interference neutralization as in the 2-user IC-CR, the whole region is achievable by time sharing, and the statement of the theorem follows.
\end{proof}

% \begin{figure}
% \centering
% \includegraphics[width=0.6\columnwidth]{DOF_region.eps}
% \caption{The DoF region of the 3-user Gaussian IC-CR with constant channel coefficients.}
% \label{DoF_reg}
% \end{figure}

\begin{remark}
Interference neutralization can also be used as a DoF achieving scheme for the time varying 2 and 3 user Gaussian IC-CR.
\end{remark}

In some special cases, the 3-user Gaussian IC-CR has 3 DoF. However, these special cases occur under conditions that do not hold almost surely, i.e. constitute a set of measure 0. This is given in the following corollary.
\begin{corollary}
\label{3UserSpecialCaseTheorem}
If the 3-user Gaussian IC-CR satisfies the following conditions,
\begin{align*}
\frac{h_{32}}{h_{31}}=\frac{h_{r2}}{h_{r1}},\quad \frac{h_{23}}{h_{21}}=\frac{h_{r3}}{h_{r1}},\quad \frac{h_{13}}{h_{12}}=\frac{h_{r3}}{h_{r2}},
\end{align*}
and
\begin{align*}
\frac{h_{11}}{h_{12}}\neq\frac{h_{r1}}{h_{r2}},\quad \frac{h_{22}}{h_{21}}\neq\frac{h_{r2}}{h_{r1}},\quad \frac{h_{33}}{h_{31}}\neq\frac{h_{r3}}{h_{r1}},
\end{align*}
then $d_1+d_2+d_3=3$.
\end{corollary}
\begin{proof}
See Appendix \ref{3UserSpecialCaseProof}.
\end{proof}

\section{Conclusion}
\label{Conc}
We studied the $K$-user Gaussian interference channel with a cognitive relay. For the 2-user case, we have obtained a new upper bound on the sum-capacity. In the general $K$-user case with time varying channel coefficients, we characterized the DoF. We have shown that while for $K=2$, the setup has 2 DoF, for $K>2$ users the DoF are upper bounded by $2K/3$. Moreover $2K/3$ DoF are achievable using interference alignment when the channels are time varying. We notice that the DoF per user is more compared to that in the $K$-user IC, where we have $1/2$ DoF per user. Thus, a cognitive relay can increase the DoF of the IC. We notice also a decrease in the per-user DoF for the $K$-user case from 1 to 2/3 as we go from $K=2$ to $K>2$. We also considered the case with constant channel coefficients, where we gave the DoF region for the 3-user case and showed that it is achievable using interference neutralization.

% \appendix
\begin{appendices}

\section{Proof of Lemma \ref{3UserICCRTheorem}}
\label{3UserICCRProof}
Let us give $(Y_1^n,m_1)$ and $(Y_1^n,m_1,m_2,\tilde{Z}^n)$ as side information to receivers 2 and 3 respectively, where $\tilde{Z}^n=(\tilde{Z}(1),\dots,\tilde{Z}(n))$ and
\begin{align*}
\tilde{Z}(i)&=Z_{2}(i)-\frac{h_{r2}(i)}{h_{r1}(i)}Z_{1}(i)\\
&\ \ -\left(Z_{3}(i)-Z_{1}(i)\frac{h_{r3}(i)}{h_{r1}(i)}\right)\frac{h_{32}(i)-\frac{h_{31}(i)h_{r2}(i)}{h_{r1}(i)}}{h_{33}(i)-\frac{h_{31}(i)h_{r3}(i)}{h_{r1}(i)}}.
\end{align*}
This random variable $\tilde{Z}$ is used to allow constructing $Y_2^n$ from $Y_3^n$, $Y_1^n$, $X_1^n$, and $X_2^n$ as we shall see next. Then, using Fano's inequality, with $\epsilon_n\to0$ as $n\to\infty$, we write
\begin{align}
&\hspace{-0.5cm}n(R_1+R_2+R_3-3n\epsilon_n)\nonumber\\
&\leq I(m_1;Y_1^n)+I(m_2;Y_2^n,Y_1^n,m_1)\nonumber\\
&\ \ +I(m_3;Y_3^n,Y_1^n,m_1,m_2,\tilde{Z}^n)\\
&= I(m_1;Y_1^n)+I(m_2;Y_1^n|m_1)+I(m_2;Y_2^n|Y_1^n,m_1)\nonumber\\
&\ \ +I(m_3;Y_1^n|m_1,m_2)+I(m_3;\tilde{Z}^n|Y_1^n,m_1,m_2)\nonumber\\
&\ \ +I(m_3;Y_3^n|Y_1^n,m_1,m_2,\tilde{Z}^n)\\
% \end{align}
% We proceed
% \begin{align}
% &\hspace{-0.5cm}n(R_1+R_2+R_3-3n\epsilon_n)\nonumber\\
&\leq I(m_1,m_2,m_3;Y_1^n)+I(m_2;Y_2^n|m_1,Y_1^n)\nonumber\\
&\ \ +I(m_3;\tilde{Z}^n|Y_1^n,m_1,m_2)\nonumber\\
&\ \ +I(m_3;Y_3^n|Y_1^n,m_1,m_2,\tilde{Z}^n)\\
\label{3UB1}
&\leq I(m_1,m_2,m_3;Y_1^n)+h(Y_2^n|m_1,Y_1^n)\nonumber\\
&\ \ -h(Y_2^n|m_1,m_2,Y_1^n)+I(m_3;\tilde{Z}^n|Y_1^n,m_1,m_2)\nonumber\\
&\ \ +h(Y_3^n|Y_1^n,m_1,m_2,\tilde{Z}^n)-h(Z_3^n|\tilde{Z}^n).
\end{align}
where we have used the chain rule and the independence of $m_1$, $m_2$ and $m_3$. Consider now the first term in \eqref{3UB1}. This is bounded by
\begin{align}
\label{3UBT1}
I(m_1,m_2,m_3;Y_1^n)\leq n\left(\frac{1}{2}\log(P)+o(\log(P))\right).
\end{align}
Moreover,
\begin{align}
\label{3UBT2}
h(Y_2^n|m_1,Y_1^n)-h(Z_3^n|\tilde{Z}^n)\leq n\left(\frac{1}{2}\log(P)+o(\log(P))\right)
\end{align}
except if $Y_2^n$ is a degraded version of $Y_1^n$ given $m_1$, which is not the case almost surely due to the randomness of the channels. Consider then the fifth term in \eqref{3UB1}, $h(Y_3^n|Y_1^n,m_1,m_2,\tilde{Z}^n)$. This can be bounded as follows
\begin{align*}
h(Y_3^n|Y_1^n,m_1,m_2,\tilde{Z}^n)&\stackrel{(a)}{=}h({Y}_3^n|{Y}_1^n,m_1,m_2,X_1^n,X_2^n,\tilde{Z}^n)\\
&\stackrel{(b)}{\leq}h(\tilde{Y}_3^n|\tilde{Y}_1^n,m_1,m_2,\tilde{Z}^n)\\
&\stackrel{(c)}{=}h\left(\hat{Y}_3^n|\tilde{Y}_1^n,m_1,m_2,\tilde{Z}^n\right)
\end{align*}
where 
\begin{itemize}
\item[$(a)$] follows since $X_1^n$ and $X_2^n$ can be constructed from $m_1$ and $m_2$, \item[$(b)$] follows by using the knowledge of $X_1^n$ and $X_2^n$ to cancel their contribution from $Y_3^n$ and $Y_1^n$, where we defined $\tilde{Y}_{3}(i)\triangleq h_{33}(i)X_{3}(i)+h_{r3}(i)X_{r}(i)+Z_{3}(i)$ and $\tilde{Y}_{1}(i)\triangleq h_{31}(i)X_{3}(i)+h_{r1}(i)X_{r}(i)+Z_{1}(i),$ and we used the fact that conditioning reduces entropy, and
\item[$(c)$] follows by the following operation
\begin{align}
\hat{Y}_{3}(i)&=\tilde{Y}_{3}(i)-\frac{h_{r3}(i)}{h_{r1}(i)}\tilde{Y}_{1}(i)\\
&=\alpha(i)X_{3}(i)+Z_{3}(i)-\frac{h_{r3}(i)}{h_{r1}(i)}Z_{1}(i),
\end{align}
where $\alpha(i)=h_{33}(i)-h_{31}(i)\frac{h_{r3}(i)}{h_{r1}(i)}\neq0$ almost surely.
% \item[$(c)$] follows by defining $Z_4^n=Z_3^n-\frac{h_{r3}}{h_{r1}}Z_1^n$.
\end{itemize}
We continue
\begin{align}
&\hspace{-0.8cm}h(Y_3^n|Y_1^n,m_1,m_2,\tilde{Z}^n)\nonumber\\
&\leq h\left(\hat{Y}_3^n|\tilde{Y}_1^n,m_1,m_2,\tilde{Z}^n\right)\\
&\stackrel{(d)}{=}h\left(\bar{Y}_3^n|\tilde{Y}_1^n,m_1,m_2,\tilde{Z}^n\right)-\frac{1}{2}\sum_{i=1}^n\log\left(\frac{\beta^2(i)}{\alpha^2(i)}\right)\\
&\stackrel{(e)}{=}h(Y_2^n|Y_1^n,m_1,m_2,\tilde{Z}^n)-\frac{1}{2}\sum_{i=1}^n\log\left(\frac{\beta^2(i)}{\alpha^2(i)}\right)\\
\label{3UB2}
&\stackrel{(f)}{\leq}h(Y_2^n|Y_1^n,m_1,m_2)-\frac{1}{2}\sum_{i=1}^n\log\left(\frac{\beta^2(i)}{\alpha^2(i)}\right),
\end{align}
where in
\begin{itemize}
\item[$(d)$] we defined $\bar{Y}_{3}(i)\triangleq\frac{\beta(i)}{\alpha(i)}\hat{Y}_{3}(i),$ with $\beta(i)=h_{32}(i)-h_{31}(i)\frac{h_{r2}(i)}{h_{r1}(i)}\neq0$ almost surely, and we used $h(aX)=h(X)+\frac{1}{2}\log(a^2)$ \cite{CoverThomas},
\item[$(e)$] follows by the constructing $Y_2(i)=\bar{Y}_3(i)+\frac{h_{r2}(i)}{h_{r1}(i)}\tilde{Y}_{1}(i)+\tilde{Z}(i)$ and reconstructing $Y_1^n$ from $(\tilde{Y}_1^n,m_1,m_2)$, and
\item[$(f)$] follows since conditioning reduces entropy.
\end{itemize}

As a result, if we consider the third and the fifth term in \eqref{3UB1} together, and use \eqref{3UB2}, we get
\begin{align}
\label{3UBT3}
h(Y_3^n|Y_1^n,m_1,m_2,\tilde{Z}^n)-h(Y_2^n|m_1,m_2,Y_1^n)\leq n\left(o(\log(P))\right).
\end{align}

Finally, the fourth term in \eqref{3UB1} satisfies
\begin{align}
\label{3UBT4}
I(m_3;\tilde{Z}^n|Y_1^n,m_1,m_2)\leq n\left(o(\log(P))\right).
\end{align}
Thus, by plugging \eqref{3UBT1}, \eqref{3UBT2}, \eqref{3UBT3}, and \eqref{3UBT4} in \eqref{3UB1}, and letting $n\to\infty$, we obtain $R_1+R_2+R_3\leq\log(P)+o(\log(P))$ and as a result, the degrees of freedom of the 3-user IC-CR is upper bounded by $d_1+d_2+d_3\leq2$.

\section{Proof of Theorem \ref{3UserSpecialCaseTheorem}}
\label{3UserSpecialCaseProof}
If $h_{32}h_{r1}=h_{r2}h_{31},$ then the upper bound in Appendix \ref{3UserICCRProof} given by  $d_1+d_2+d_3\leq2$ does not hold since $\beta=0$. It can be similarly shown that, by giving similar side information as in Appendix \ref{3UserICCRProof} to receivers 1 and 3, and 1 and 2, the conditions ${h_{23}}{h_{r1}}={h_{r3}}{h_{21}}$, and ${h_{13}}{h_{r2}}={h_{r3}}{h_{12}}$, are required so that the DoF does not collapse to 2. Now, as long as
\begin{align*}
\frac{h_{11}}{h_{12}}\neq\frac{h_{r1}}{h_{r2}},\quad \frac{h_{22}}{h_{21}}\neq\frac{h_{r2}}{h_{r1}},\quad \frac{h_{33}}{h_{31}}\neq\frac{h_{r3}}{h_{r1}},
\end{align*}
the relay can cancel the interference at all receivers simultaneously, and thus 3 DoF are achievable.

\end{appendices}

\bibliography{myBib}

\end{document}